\documentclass[11pt,letterpaper]{article}
\usepackage[margin=1in]{geometry}
\usepackage{mathpazo}
\usepackage[T1]{fontenc}
\usepackage{fix-cm}   % scalable CM at all sizes
\usepackage{lmodern}  % Latin Modern — Type 1 outlines at all sizes (needed for \cmark nodes)
\usepackage{microtype}
\usepackage{amsmath,amssymb,amsthm}
\usepackage{listings}
\usepackage{xcolor}
\usepackage{hyperref}
\hypersetup{
  pdftitle={Nidus: Externalized Reasoning for AI-Assisted Engineering},
  pdfauthor={Danil Gorinevski},
  pdfsubject={AI governance, formal verification, LLM agents},
  pdfkeywords={constraint surface, proof obligations, SMT, S-expressions, multi-agent coordination, V-model, guidebooks},
  pdfcreator={LaTeX with hyperref},
  pdfproducer={pdfTeX},
}
\usepackage{tikz}
\usetikzlibrary{arrows.meta,positioning,shapes.geometric,automata,fit,backgrounds}
\usepackage{booktabs}
\usepackage{enumitem}
\usepackage{pgfplots}
\pgfplotsset{compat=1.18}
\usepackage{graphicx}
\usepackage[font=small,labelfont=bf,skip=6pt]{caption}
\usepackage{pifont}
\newcommand{\cmark}[1]{\tikz[baseline=(c.base)]{\node[circle,fill=black,text=white,inner sep=0pt,minimum size=1.1em,font=\usefont{T1}{lmss}{bx}{n}\tiny](c){#1};}}

\theoremstyle{definition}
\newtheorem{definition}{Definition}
\newtheorem{theorem}{Theorem}
\newtheorem{property}{Property}
\theoremstyle{remark}
\newtheorem*{remark}{Remark}

\title{% 
  \textbf{Nidus: Externalized Reasoning\\for AI-Assisted Engineering}
}
\author{
Danil Gorinevski\thanks{cybiont GmbH, Kantonsstrasse 17, 8862 Sch\"ubelbach, Switzerland.  ORCID: \href{https://orcid.org/0009-0007-5992-6667}{0009-0007-5992-6667}} \\
\texttt{research@cybiont.com}
}
\date{April 5, 2026}
  
\begin{document}
\maketitle
  
\begin{abstract}
We present Nidus, a governance runtime that mechanizes the V-model for AI-assisted software delivery. 
 In the self-hosting deployment, three LLM families (Claude, Gemini, Codex) delivered a 100{,}000-line system under proof obligations verified against the current obligation set on every commit.  The system governed its own construction.

Engineering invariants---traced requirements, justified architecture, evidenced deliveries---cannot be reliably maintained as learned behavior; assurance requires enforcement by a mechanism external to the proposer.  Nidus externalizes the engineering methodology into a decidable artifact verified on every mutation before persistence.  Organizational standards compile into \emph{guidebooks}---constraint libraries imported by governed projects and enforced by decidable evaluation.

Four contributions: (1)~\emph{recursive self-governance}---the constraint surface constrains mutations to itself; (2)~\emph{stigmergic coordination}---friction from the surface routes agents without central control; (3)~\emph{proximal spec reinforcement}---the living artifact externalizes the engineering context that RL and long-chain reasoning try to internalize; the specification is the reward function, UNSAT verdicts shape behavior at inference time, no weight updates; (4)~\emph{governance theater prevention}---compliance evidence cannot be fabricated within the modeled mutation path.  The constraint surface compounds: each obligation permanently eliminates a class of unengineered output.  The artifact's development history is a formal development---every state satisfies all active obligations, and the obligation set grows monotonically.
\end{abstract}

%% ====================================================================
\section{Introduction}
\label{sec:introduction}

Engineering is not code.  Engineering is the chain of invariants that makes code trustworthy: every requirement traced, every architecture decision justified, every delivery evidenced.  The V-model---require\-ments $\to$ architecture $\to$ design $\to$ verification $\to$ evidence---encodes these invariants.  It remains the standard in safety-critical domains (ISO~26262, DO-178C, IEC~62304).

LLM agents produce code with more defects~\cite{coderabbit2025} and no traceability.  The problem is not capability.  RLVR~\cite{lightman2023} and process reward models show that training can improve adherence to external verifiers---but the verifier remains the source of assurance, not the trained policy.  Engineering invariants cannot be reliably \emph{maintained} as learned behavior; they must be \emph{enforced} by a mechanism external to the proposer.  Training data captures what engineers did; it does not capture the rules they maintained.  A model that ``knows'' traceability from textbooks will approximate it, and the approximation degrades under pressure---an agent fabricates evidence (Section~\ref{sec:theater}), another skips synthesis gates because ``I optimize for speed over correctness.''  You do not train a compiler to reject type errors; you build a type checker.

Test-driven development (TDD) bridges code correctness.  Planning layers attempt to bridge the rest: agent frameworks generate implementation plans, add stateful workflows, add persistent state.  These are useful---but in each case the model proposes the plan and the model verifies the plan.  For safety-critical invariants, self-assessment is not sufficient; an external, decidable mechanism is required.

Nidus externalizes the methodology into a decidable artifact.  Requirements, architecture, design, traces, proofs, and evidence live in a single object---a \emph{living artifact}---verified on every mutation before persistence.  Organizational standards compile into \emph{guidebooks}---constraint libraries any project imports.  The methodology is in the environment, enforced by decidable evaluation; Z3 validates constraint composition at guidebook authoring time.

The architect provides what models cannot: judgment about \emph{what} to build and \emph{how} to structure it.  Requirements, architecture decisions, design elements, and the development vector are human input.  Three LLM families (Claude, Gemini, Codex) delivered the reference implementation through this governance loop.

The reference implementation uses S-expressions because they \emph{are} SMT-LIB2~\cite{smtlib2}: a solver-aligned representation that minimizes translation between the artifact and the solver~\cite{demoura2008}.  The tools for decidable verification---S-expressions~\cite{mccarthy1960}, SMT solvers, decidable logics---have existed for decades.  What changed is the rate of the proposer.

We call the union of proof obligations, guidebook constraints, and monotonic invariants the \emph{constraint surface} (formally defined in Section~\ref{sec:kernel}).  Four contributions: (1)~\emph{recursive self-governance}---the surface constrains mutations to itself (Section~\ref{sec:self-governance}); (2)~\emph{stigmergic coordination}---friction from the surface routes agents without an orchestrator (Section~\ref{sec:stigmergy}); (3)~\emph{proximal spec reinforcement}---the artifact externalizes what RL and planning layers try to internalize (Section~\ref{sec:psrl}); (4)~\emph{enforcement mechanisms}---governance theater prevention, SOP gate sequences, lesson provenance (Section~\ref{sec:theater}).

%% ====================================================================
\section{Nidus: A Decidable Living Specification}
\label{sec:language}

What does the living artifact look like in practice?  Nidus replaces codebases with a \emph{living specification}: the S-expression \emph{is} the design, the proof obligation set, and the governance authority.  Implementation code, RTL, and documentation are governed through anchor paths (\texttt{code-paths}, \texttt{test-paths}) and proof obligations.

\subsection{Representational Closure}
\label{sec:repr-closure}

The living artifact serves four roles simultaneously: the \textbf{database} (for the daemon that enforces mutations), the \textbf{solver input} (for Z3 that checks proof obligations), the \textbf{agent context} (an LLM ingests the entire artifact in one read), and the \textbf{specification} (a human reads the same object the solver checks).  We call this \textbf{representational closure}~\cite{mccarthy1960}: every participant---human, LLM, solver---reads, writes, and reasons over the identical object.

In traditional engineering, this state is fragmented: requirements in Jira, architecture in Confluence, traces in spreadsheets, test results in CI logs.  The living artifact unifies them---primarily for the AI agents, which need the full engineering state in their context window.  An agent that cannot see the traces cannot maintain them; an agent that cannot see the proof obligations cannot satisfy them.  The invariant chain becomes a structural property of the object, not a behavioral property of whoever writes to it.

\subsection{Annotated Example}
\label{sec:example}

Circled numbers \cmark{1}--\cmark{12} mark every structural element; Section~\ref{sec:formal} maps each to a mathematical object.

\begin{lstlisting}[caption={Annotated Nidus artifact.},escapechar=@]
(nidus-system "Interceptor"                     @\cmark{1}@
  (requirements                                  @\cmark{2}@
    (req UR-01 (kind sovereignty) (source RAD)
      (shall "All classification executes locally")
      (constraint (never (payload-egress data))))
    (req FR-01 (kind functional) (source DDD)
      (shall "Verdict is one of ALLOW BLOCK WARN")))
  (architecture                                  @\cmark{3}@
    (component Interceptor
      (responsibility "Intercept and route traffic"))
    (component Brain
      (responsibility "Local classification"))
    (connector Interceptor Brain               @\cmark{4}@
      (flow CLASSIFY)
      (protocol synchronous)))
  (design                                        @\cmark{5}@
    (value-object Verdict (values ALLOW BLOCK WARN)))
  (workflows                                     @\cmark{6}@
    (workflow classify-flow
      (states idle classifying decided)
      (initial idle)
      (transition idle classifying
        (guard (> priority 0)))
      (transition classifying decided
        (guard (< elapsed 100)))))
  (features                                      @\cmark{7}@
    (feature FEAT-01
      (name "Local classification engine")
      (status claimed)
      (scope (requirements UR-01 FR-01)
        (code-paths "src/brain.py"))))
  (traceability                                  @\cmark{8}@
    (trace TR-01 UR-01 Brain Verdict)
    (trace TR-02 FR-01 Brain Verdict))
  (proof-obligations                             @\cmark{9}@
    (proof PO-TRACE-01
      (kind traceability-complete)
      (description "Every requirement has a trace"))
    (proof PO-CONN-01
      (kind connector-integrity)
      (description "Connectors reference components"))
    (proof PO-DAG-01
      (kind dag-enforcement)
      (description "No dependency cycles"))
    (proof PO-WF-01
      (kind workflow-satisfiability)
      (description "Workflow guards are satisfiable")))
  (coordination                                  @\cmark{10}@
    (claim (agent opus-a1b2) (feature FEAT-01)
      (lease-expires "2026-03-20T14:00:00Z")))
  (guidebooks                                    @\cmark{11}@
    (imports "../epoch.guidebook.epoch")))

;; Sidecar file: .epoch.friction.epoch          @\cmark{12}@
(agent-friction-ledger
  (events
    (friction-event FX-01
      (kind agent_rejection)
      (agent gemini-8536)
      (timestamp "2026-03-09T21:42:03Z"))))
\end{lstlisting}

\noindent A \textbf{connector} \cmark{4} is a typed data-flow edge.  A \textbf{workflow} \cmark{6} is an FSM with QF\_LIA guards.  \textbf{Features} \cmark{7} carry scope: which requirements, code-paths, and test-paths belong to the delivery.  A \textbf{trace} \cmark{8} links requirement $\to$ component $\to$ design element.  \textbf{Proof obligations} \cmark{9} are predicates over the artifact's finite sections, evaluated on every mutation.  Each has a \texttt{kind} that names a \emph{verification category}---\texttt{traceability-complete}, \texttt{connector-integrity}, etc.---and binds to an evaluator function in the kernel.  Obligations span structural, schema, SMT, temporal, and domain-specific kinds.  Two layers of DAG enforcement illustrate the distinction: the local \texttt{PO-DAG-01} (kind \texttt{dag-enforcement}) checks that \emph{components} form an acyclic dependency graph, while the inherited \texttt{GC-MODULE-DAG} (kind \texttt{call-graph-dag}) checks that \emph{Lisp modules} in the daemon codebase have no circular imports.  Both are DAG checks, but over different sections of the artifact.

\emph{Example.}  An agent adds requirement \texttt{FR-02} without a corresponding trace.  \texttt{PO-TRACE-01} (\texttt{traceability-complete}) rejects: FR-02 has no trace link.  The agent adds \texttt{(trace TR-03 FR-02 Brain Verdict)}.  All POs pass; persisted.  The loop: propose, verify, repair, commit.

\textbf{Coordination} \cmark{10}: claims with leases.  \textbf{Friction ledger} \cmark{12}: append-only sidecar for agent-quality signals.

\textbf{Guidebook imports} \cmark{11}.  The Interceptor project defines four local proof obligations (PO-TRACE-01 through PO-WF-01).  But line~\cmark{11} imports an organizational guidebook.  That guidebook carries additional constraints---engineering standards the project inherits without defining them locally:

\begin{lstlisting}[caption={Guidebook constraints inherited via \texttt{(imports ...)} in Listing~1.},escapechar=@,basicstyle=\ttfamily\footnotesize]
;; From epoch.guidebook.epoch (imported by Interceptor)

(guidebook-constraint GC-SCOPE-COMPLETENESS       @\cmark{A}@
  (z3_formula (implies feature_delivered           @\cmark{B}@
               (and has_code_paths has_test_paths
                    has_requirements)))
  (po-kind feature-code-test-symmetry)             @\cmark{C}@
  (description "Delivered features must have
    code-paths, test-paths, and requirements"))

(guidebook-constraint GC-MODULE-DAG               @\cmark{A}@
  (z3_formula (and (not parse_imports_query)       @\cmark{B}@
                   (not mutate_imports_query)
                   (not verify_imports_query)))
  (po-kind call-graph-dag)                         @\cmark{C}@
  (description "Module DAG: no back-edges"))

(guidebook-constraint GC-EVIDENCE-PROVENANCE      @\cmark{A}@
  (z3_formula (implies evidence_submitted          @\cmark{B}@
               (and witness_registered hash_present)))
  (po-kind evidence-provenance)                    @\cmark{C}@
  (description "Evidence requires registered
    witness and server-computed hash"))
\end{lstlisting}

\noindent \cmark{A}~Constraint identity. \cmark{B}~SMT formula: compiled to native Z3 Boolean assertions (Section~\ref{sec:kernel} details the compilation). \cmark{C}~PO-kind binding: links to an evaluator function that iterates artifact sections and returns violations.

\emph{Terminology note: \texttt{kind} vs.\ \texttt{po-kind}.}  Inside the artifact (Listing~1), a proof obligation's \texttt{kind} names the verification category it belongs to---e.g.\ \texttt{traceability-complete}.  Inside a guidebook constraint (Listing~2), the \texttt{po-kind} field serves the same purpose: it identifies the evaluator function the kernel dispatches.  The two are the same namespace; the separate keyword exists because a guidebook constraint also carries a \texttt{z3\_formula}, whereas a local PO relies solely on its evaluator.

Each constraint has two enforcement paths.  The \texttt{po-kind} evaluator traverses the artifact's finite sections and returns violations or nil.  The \texttt{z3\_formula} compiles to a Z3 assertion checked for satisfiability.  Many constraints use both: the evaluator performs the structural check; the formula provides a declarative specification the solver verifies independently.

\emph{Why guidebooks matter.}  Without the import on line~\cmark{11}, the Interceptor's four local POs would accept a delivery with \texttt{code-paths} but no \texttt{test-paths}---traceability is complete, connectors are valid, the DAG is acyclic, and the workflow guards are satisfiable.  \emph{With} the import, the organizational guidebook adds \texttt{GC-SCOPE-COMPLETENESS}, which requires \texttt{test-paths} on every delivered feature.

Agent \texttt{opus-a1b2} attempts to deliver FEAT-01 (Listing~1, currently \texttt{claimed}).  The mutation changes status to \texttt{delivered}.  The daemon evaluates all obligations---local and inherited---against the proposed state:

\begin{lstlisting}[caption={Verification output: local POs pass, inherited constraint rejects.},basicstyle=\ttfamily\footnotesize]
commit_change_set: verify FEAT-01 delivery
  PO-TRACE-01  traceability-complete     PASS
  PO-CONN-01   connector-integrity       PASS
  PO-DAG-01    dag-enforcement           PASS
  PO-WF-01     workflow-satisfiability   PASS
  -- inherited from epoch.guidebook.epoch --
  GC-MODULE-DAG       call-graph-dag          PASS
  GC-EVIDENCE-PROV.   evidence-provenance     PASS
  GC-SCOPE-COMPLETE.  feature-code-test-sym.  FAIL
    FEAT-01: has code-paths, missing test-paths
    z3: (implies feature_delivered
          (and has_code_paths has_test_paths
               has_requirements))  -> UNSAT
REJECTED: 1 violation (GC-SCOPE-COMPLETENESS)
\end{lstlisting}

\noindent The agent adds \texttt{(test-paths "tests/test\_brain.py")} to the scope and retries; all obligations pass; FEAT-01 is delivered.  The organizational standard was enforced mechanically---the Interceptor project never defined this constraint.

Inheritance is monotonic (Definition~5): $\Pi(G_{\mathrm{parent}}) \subseteq \Pi(G_{\mathrm{child}})$.  Three layers compose: \emph{code practices} (DAG, function length, naming), \emph{engineering process} (scope completeness, TDD evidence, claim-before-work), and \emph{anti-theater} (evidence provenance, claim-before-dispatch).  Every governed project inherits all three---the methodology is a \emph{reusable proof library}.

\paragraph{Guidebooks as domain introduction.}  Guidebooks carry more than constraints.  A \texttt{(modules~\ldots)} declaration names domain-specific Lisp modules that the daemon loads on demand when any project imports the guidebook.  The chip-design guidebook declares four modules---an RTL parser, five hardware PO evaluators, a synthesis checker, and a simulation evidence gate---alongside eight constraints and onboarding templates.  A single \texttt{(imports "chip-design.guidebook.epoch")} line gives the verilog project a complete hardware governance stack: constraint enforcement, domain-specific verification, and tool integration.  The kernel does not change.  Adding a new domain (VHDL, Chisel, a compliance framework) means writing a guidebook and its modules---not modifying the daemon.  Each module must export a set of named evaluator functions; every evaluator accepts the artifact's parsed S-expression tree plus the mutated section, and returns either \texttt{nil} (pass) or a list of violation records.  The kernel validates this interface at load time.  This is the extensibility mechanism: guidebooks are not just constraint libraries; they are \emph{pluggable domain packages}.

\subsection{Formal Definition}
\label{sec:formal}

\begin{definition}[Nidus System \cmark{1}]
$\mathcal{S} = (R, A, D, W, \mathcal{F}, \mathcal{T}, \Pi, \Sigma, \Lambda, G)$ where:
$R$ \cmark{2} requirements, $A = (C, E)$ \cmark{3}\cmark{4} architecture (components + connectors), $D$ \cmark{5} design elements, $W$ \cmark{6} workflows (FSMs with QF\_LIA guards), $\mathcal{F}$ \cmark{7} features (scope + delivery status), $\mathcal{T}$ \cmark{8} trace links, $\Pi$ \cmark{9} proof obligations, $\Sigma = (K, L)$ \cmark{10}\cmark{12} coordination state (claims + friction ledger), $\Lambda$ lessons (failure $\to$ root-cause $\to$ obligation; see Section~\ref{sec:theater}), $G$ \cmark{11} guidebook hierarchy.
\end{definition}

\begin{definition}[Trace Link \cmark{8}]
$t = (r, c, d)$: requirement $r \in R$, component $c \in C$, design element $d \in D \cup W$.
\end{definition}

\begin{property}[Decidability and Complexity \cmark{6}\cmark{9}]
All sets are finite; workflow guards are QF\_LIA/QF\_IDL~\cite{smtlib2}.  Obligations decompose into three decidable categories:
\begin{enumerate}[nosep,label=(\alph*)]
  \item finite graph checks (traceability, orphans, DAG acyclicity): $O(|C|{+}|E|{+}|\mathcal{T}|)$,
  \item schema checks (field completeness): $O(|\mathcal{F}|)$,
  \item bounded arithmetic over workflow guards, decided by Z3 in time polynomial in the number of guard variables (fixed per workflow).
\end{enumerate}
Total verification cost per mutation: $O(|\Pi| \cdot |\mathcal{S}|)$ where $|\mathcal{S}|$ is the total number of entries across all artifact sections, linear in artifact size for each obligation.
\end{property}

\begin{definition}[Immutable Obligations]
A proof obligation $\pi$ with \texttt{(immutable t)} belongs to $\Pi^{\mathrm{imm}} \subseteq \Pi$.  Mutations that remove or weaken any $\pi \in \Pi^{\mathrm{imm}}$ are rejected.  Theorem~\ref{thm:correctness} depends on this: $\Pi_0^{\mathrm{imm}} \subseteq \Pi_n$ for all accepted histories.
\end{definition}

\begin{definition}[Trust Tier \cmark{12}]
$\textit{tier}(a, L) = \textit{unrestricted}$ if $\textit{fric}(a, L, w) < \theta_1$; $\textit{supervised}$ if $\theta_1 \leq \textit{fric} < \theta_2$; $\textit{restricted}$ if $\textit{fric} \geq \theta_2$.  Restricted agents cannot claim features.  Clean deliveries earn promotion.
\end{definition}

\begin{definition}[Guidebook Hierarchy \cmark{11}]
$G = (G_1, \ldots, G_m)$: constraint bundles imported transitively.  Each $G_i$ declares SMT formulas compiled to Z3 assertions.  Inheritance is monotonic: $\Pi(G_{\mathrm{parent}}) \subseteq \Pi(G_{\mathrm{child}})$.
\end{definition}

\begin{definition}[Governance Loop \cmark{9}\cmark{10}\cmark{12}]
\label{def:psrl}
(1)~Agent observes $\mathcal{S}$. (2)~Proposes mutation $m$. (3)~Kernel evaluates $m$ against all $\pi \in \Pi$: pass $\to$ persist + \texttt{probe\_success}; fail $\to$ reject + \texttt{agent\_rejection} + UNSAT core (the minimal subset of constraints that caused the failure).  (4)~Tier recomputed.
\end{definition}

\begin{figure}[htbp]
\centering
\includegraphics[width=\textwidth]{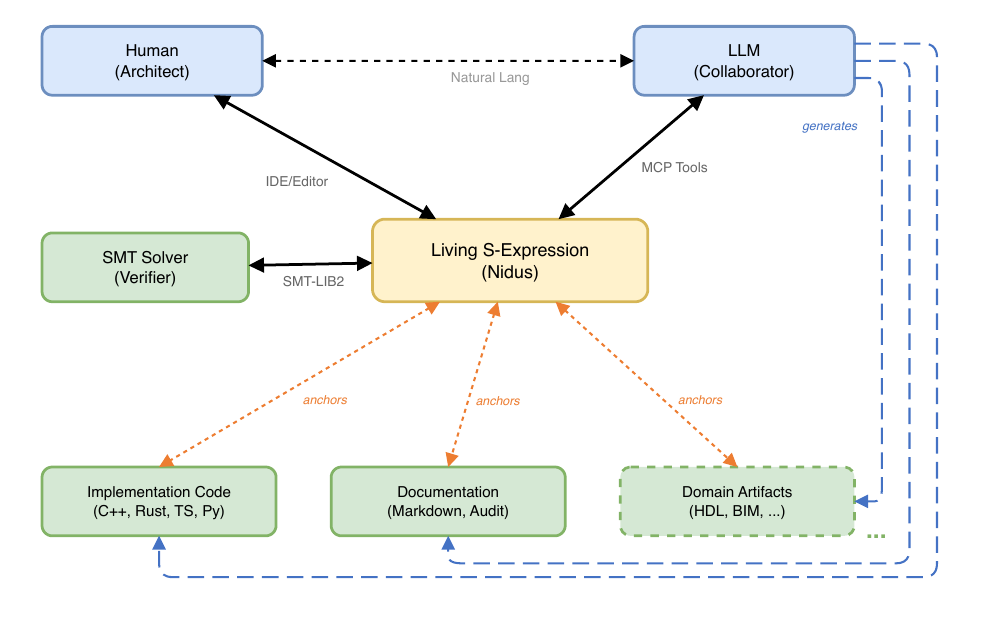}
\caption{Representational closure.  Human, LLM, and Solver read/write the same S-expression.}
\label{fig:quadrivium}
\end{figure}

%% ====================================================================
\section{The Verification Kernel}
\label{sec:kernel}

\paragraph{Guidebook constraint compilation.}  Guidebook constraints (Listing~2) have two enforcement paths.  A \texttt{po-kind} field binds to an evaluator function registered in the kernel's dispatch table---evaluators are not hardcoded; guidebook modules can register new evaluator functions at import time (Section~\ref{sec:language}), extending the kernel without modifying it; the evaluator traverses artifact sections and returns violations.  A \texttt{z3\_formula} field is compiled to a native Z3 assertion: propositional variables are declared as Z3 Booleans, operators map to their SMT-LIB2 equivalents, and the assertion is checked for satisfiability.  Internally inconsistent formulas (e.g.\ \texttt{(and p (not p))}) are flagged as UNSAT at import time.  The compilation is mechanical---the S-expression grammar is a strict subset of SMT-LIB2.  Propositional variables map to artifact predicates: \texttt{feature\_delivered} is true when any feature's status equals \texttt{delivered}; \texttt{has\_test\_paths} is true when its scope contains a \texttt{test-paths} field.  The evaluator computes these predicates from the artifact state before each Z3 call.

\paragraph{Obligation taxonomy.}  Proof obligations span seven categories: \emph{structural}, \emph{schema}, \emph{SMT}, \emph{temporal}, \emph{anti-theater}, \emph{coverage}, and \emph{domain} (including hardware-specific kinds: clock-domain crossing, port-width consistency).

Guards compile to SMT-LIB2 with minimal translation overhead; Z3~\cite{demoura2008} consumes them directly.  Graph checks are $O(|V|+|E|)$ where $V$ is the set of components and $E$ the set of connectors in the architecture; SMT queries over QF\_LIA with fixed variable counts.  Failed mutations return an UNSAT core; \texttt{verified\_repair\_search} returns only proven fixes.

\texttt{what\_if} applies a proposed mutation to an in-memory copy of the artifact, runs all proof obligations against the copy, and discards it---the agent receives a pass/fail verdict and any UNSAT core without risking state corruption or accumulating friction.  Agents use \texttt{what\_if} to explore the constraint landscape at zero risk.  Monotonic invariants ensure requirement count, PO count, and delivered status can never decrease.  The artifact is a single object---rather than a distributed collection of files---because representational closure (Section~\ref{sec:repr-closure}) requires that human, LLM, and solver all read the same structure; splitting it would reintroduce the fragmentation the design eliminates.  Git is the artifact's write-ahead log.

\paragraph{Verification before persistence.}  The practical consequence is that \emph{no other agent or tool ever observes a state that fails the active verification gate}.  In conventional CI/CD, a broken commit exists on the branch---other agents may pull it, build on it, or reference it before the pipeline reports a failure.  Nidus eliminates this window: the daemon verifies the mutation in memory and only then appends it to the git write-ahead log.  The repository contains exclusively states that have passed the active verification gate.  Theorem~\ref{thm:correctness} formalizes this: every persisted $\mathcal{S}_i$ satisfies $\Pi_i$.  Verification is sound with respect to the currently encoded obligation set and trusted evaluators, but does not guarantee properties that have not been modeled.  A proof obligation that runs the feature's test suite during delivery extends this guarantee to the code: nothing ships unless it compiles and passes.  Kleppmann~\cite{kleppmann2025} predicts that AI will make formal verification mainstream; Nidus is an instance.

\begin{figure}[htbp]
\centering
\includegraphics[width=\textwidth]{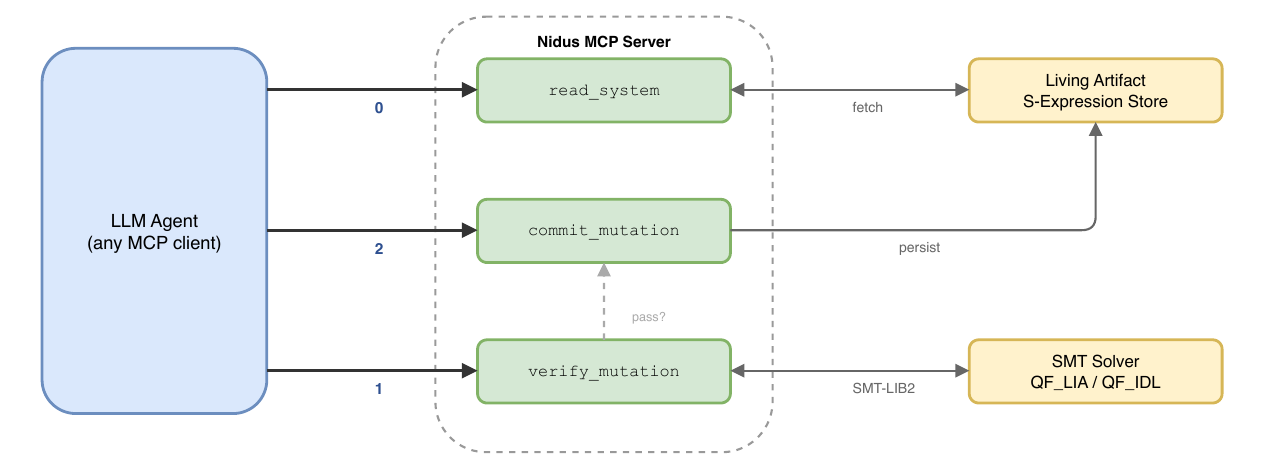}
\caption{MCP server: read (step~0), verify (step~1), commit (step~2).  Commit is gated on verification.}
\label{fig:mcp}
\end{figure}

\subsection{Formal Guarantees}

The following properties are straightforward by construction, but we state them explicitly because the rest of the paper refers to them as contracts.

\begin{theorem}[Decidable Verification]
For any $\mathcal{S}$ and mutation $m$, verification of $m$ against $\Pi$ terminates in bounded time.
\end{theorem}

\begin{proof}
All sets finite (Def.~1); graph checks $O(|V|+|E|)$; guards QF\_LIA (decidable); $\Pi$ finite.
\end{proof}

\begin{theorem}[Incremental Correctness]
\label{thm:correctness}
If $\mathcal{S}_0$ satisfies $\Pi_0$ and mutations $m_1, \ldots, m_n$ are accepted, then every $\mathcal{S}_i$ satisfies $\Pi_i$ and $\Pi_0^{\mathrm{imm}} \subseteq \Pi_n$.
\end{theorem}

\begin{proof}
Induction.  Accept only if $\mathcal{S}_i$ satisfies all $\pi \in \Pi_i$.  Immutable POs cannot be removed.
\end{proof}

\begin{theorem}[Monotonic Governance]
Under guidebook hierarchy $G$, accepted child mutations satisfy all parent obligations.
\end{theorem}

\begin{proof}
$\Pi(G_{\mathrm{parent}}) \subseteq \Pi(G_{\mathrm{child}})$.  Child verification checks all.
\end{proof}

\begin{theorem}[Engineering Record Completeness]
\label{thm:record}
If the artifact history $H = (\mathcal{S}_0, m_1, \ldots, m_n, \mathcal{S}_n)$ satisfies Theorem~\ref{thm:correctness} and each accepted mutation carries an attestation (feature, agent, verification result, fingerprint transition), then compliance certificates, traceability matrices, and impact analyses are derivable as projections of~$H$ without external state.
\end{theorem}

\begin{proof}
Each $\mathcal{S}_i$ contains all requirements, traces, proof obligations, evidence, and coordination entries (Definition~1).  Each attestation records provenance.  $H$ is totally ordered by git.  Any required report is a query over~$H$.
\end{proof}

\noindent Theorem~\ref{thm:correctness} guarantees that no accepted mutation violates the current PO set.  This eliminates regressions \emph{that the PO set captures}.  Semantic defects not encoded as POs---performance regressions, emergent cross-artifact invariants, specification errors in the POs themselves---are not blocked.  The surface is sound but not complete; its coverage grows as the architect crystallizes new obligations from observed failures.

We call $\Pi \cup \Pi(G) \cup \{\text{monotonic ratchets}\}$ the \textbf{constraint surface}: the union of artifact proof obligations, inherited guidebook constraints, and the invariants that prevent requirement/PO count from decreasing.  The surface is decidable (Theorem~1), preserves correctness (Theorem~2), and composes monotonically across organizations (Theorem~3).

\begin{remark}[Computational Separation]
Nidus is the bounded half of a Turing-complete system.  The proposer (LLM) explores an unbounded search space; the kernel gates every persistent transition through finite, terminating verification.  This separation means that adding agents does not increase verification cost, and when the model changes, the constraint surface remains sound.
\end{remark}

%% ====================================================================
\section{Recursive Self-Governance}
\label{sec:self-governance}

\begin{definition}[Recursive Self-Governance]
A Nidus system exhibits \emph{recursive self-governance} if mutations to $\Pi$ are themselves subject to verification by~$\Pi$.
\end{definition}

ACL2~\cite{acl2} does not re-check theorem database changes; seL4~\cite{klein2009} requires human-guided Isabelle/HOL; Cedar~\cite{cedar2024} analyzes policies but is not itself a policy.  Nidus: governance rules live in the same artifact as governed content, every mutation---including to $\Pi$---passes through the same automatic gate.

\texttt{PO-CONFORMANCE} verifies artifact structure; \texttt{(immutable t)} POs block removal.  TCB: the kernel evaluator.  Governed surface: the entire artifact.  Section~\ref{sec:evaluation} reports the empirical test: add PO, violate it, attempt removal.

\paragraph{Runtime self-governance.}  Recursive self-governance extends to the runtime itself.  The kernel evaluates proof obligations not only on the artifact but on its own process state: a memory-ceiling PO rejects commits when heap usage exceeds a configurable threshold (derived from the 160\,GB OOM crashes in Section~\ref{sec:theater}), a latency-bound PO flags verification times exceeding 5\,s, and a state-freshness PO detects when the in-memory artifact has diverged from the WAL.  These are ordinary proof obligations in the artifact---an agent can read them, and the kernel enforces them on itself.  Relative to the systems reviewed here (Section~\ref{sec:related}), we are not aware of another architecture that represents runtime self-constraints as first-class artifact obligations enforced by the same mutation gate.

\paragraph{Autonomic evolution.}  The \emph{decay detector} is a background process that scans the friction ledger for recurring failures, grouped by PO kind (e.g.\ repeated \texttt{traceability-complete} rejections).  When a particular PO-kind's failure count exceeds a configurable threshold within a rolling window, the system generates a candidate proof obligation, tests it via \texttt{what\_if} against the current artifact, and deposits it as an open feature for an agent to commit.  Concrete instance: agents submitted features with empty scopes.  The system generated \texttt{PO-DELIVERY-CASCADE}---a proof obligation that rejects any feature delivery whose scope lacks at least one requirement, one code-path, and one test-path---verified it, and an agent committed it.  The vulnerability is permanently closed.

%% ====================================================================
\section{Stigmergic Multi-Agent Coordination}
\label{sec:stigmergy}

Friction accumulates unevenly---not because one agent is worse, but because its trajectory diverged into a constraint-dense neighborhood.  Trust tiers, derived as a pure function over the ledger's rolling window, route unproductive trajectories toward uncontested work.  This is a \textbf{throughput mechanism}, not a quality judgment---mirroring ant colony reallocation~\cite{theraulaz1998}: the ledger \emph{is} the pheromone; the tier gate \emph{is} the pivot.

$P(\textit{success})$ for each (agent, feature) pair cross-references per-PO-kind reputation vectors against scope-implied PO requirements.  An agent with high traceability accuracy but low evidence accuracy is routed to traceability-heavy features---not gated entirely.  Multi-level identity: \emph{session} (CAS), \emph{family} (blind spot detection), \emph{trajectory} (routing).  Success templates from other families are surfaced in workbenches; pre-flight warnings fire on historically failed PO kinds.  When multiple agents contribute to the same area, collaborative credit accrues to all contributors---the game rewards cooperation, not competition.

\begin{figure}[htbp]
\centering
\includegraphics[width=\textwidth]{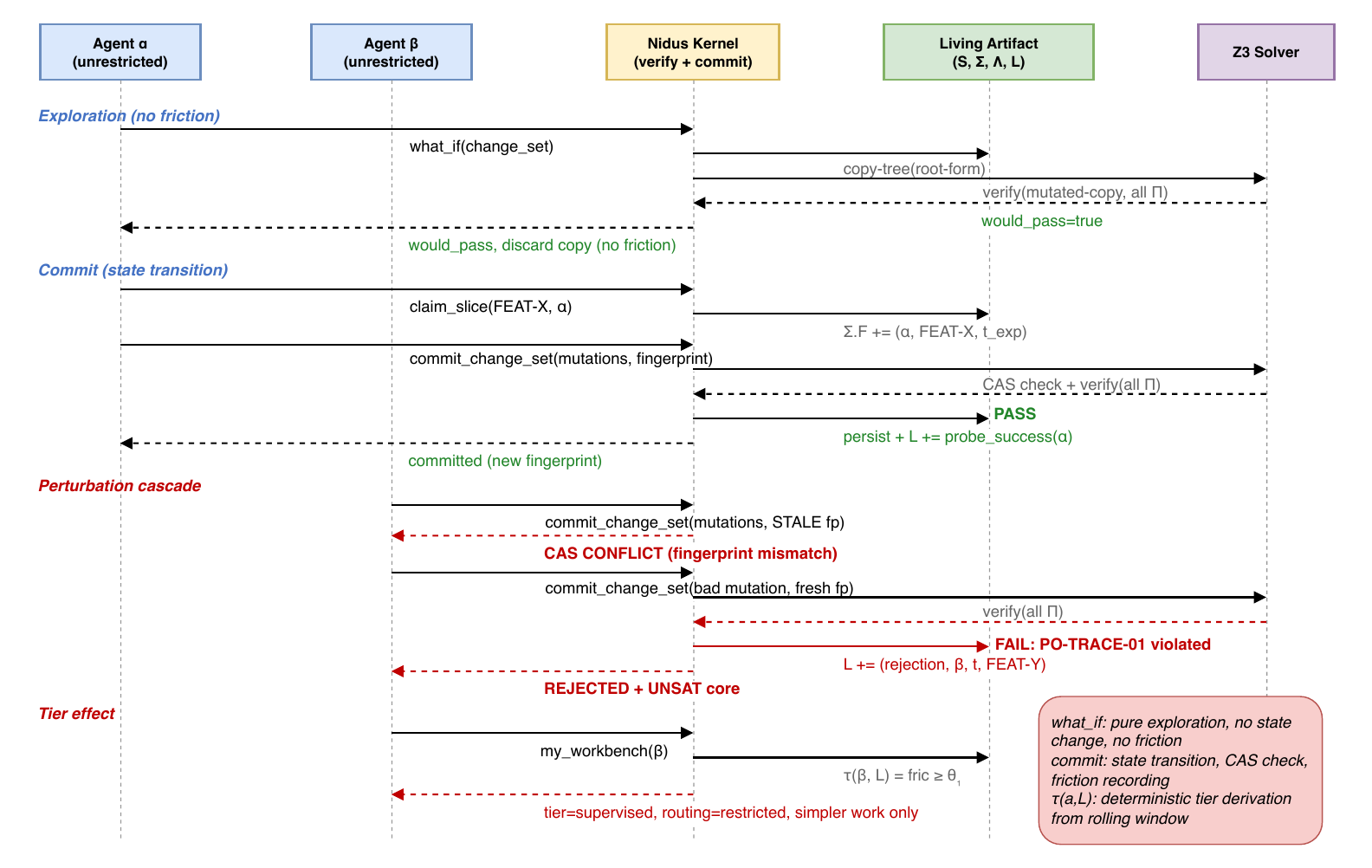}
\caption{Stigmergic dynamics.  Phase~1: speculative exploration.  Phase~2: verified commit.  Phase~3: rejection + UNSAT core.  Phase~4: tier narrowing.  No orchestrator.}
\label{fig:stigmergy}
\end{figure}

%% ====================================================================
\section{Proximal Spec Reinforcement (PSR)}
\label{sec:psrl}

Section~\ref{sec:introduction} argued that engineering invariants require external enforcement.  PSR is the mechanism.  The V-model chain---requirements $\to$ architecture $\to$ design $\to$ traces $\to$ proofs $\to$ evidence---lives in the verified artifact, not in model parameters.  When verification fails, the kernel returns an \emph{UNSAT core}---the minimal subset of proof obligations whose conjunction makes the proposed state unsatisfiable---which tells the agent exactly which link in the engineering chain broke.  The friction tier constrains which features the agent may attempt next.

Human engineering teams succeed because the \emph{environment}---architectural plans, requirements, design decisions---constrains what individuals produce.  AI agents without a persistent engineering context have no plan to be better \emph{toward}.  RL~\cite{lightman2023} makes each generation better; PSR gives it direction.  The architect's contribution is that direction: not just proof obligations (what must hold) but architecture decisions (how to structure it), design elements (what abstractions to use), and the development vector (what to build next).  Guidebooks are organizational knowledge---crystallized from architectural judgment, inherited by every project that imports them.  The constraint surface is the reward model---but decidable, explicit, and persistent across sessions and model generations.

In practice, agents from all three families produced verified deliveries within their first session---the UNSAT core narrowed the search space on each rejection, providing the equivalent of a curriculum without any prior training on the artifact's structure.

The mapping to reinforcement learning is structural:

\begin{center}
{\small
\begin{tabular}{@{}ll@{}}
\toprule
\textbf{RL concept} & \textbf{Nidus mechanism} \\
\midrule
Reward signal & \texttt{what\_if} returns accept/reject + UNSAT core before persistence \\
Reward shaping & Friction score: sum of per-PO-kind failure weights (graded, not binary pass/fail) \\
Policy gradient & Repair templates from UNSAT cores: concrete S-expression fixes \\
Value function & Per-PO-kind reputation vectors route agents to their strengths \\
Curriculum & Pre-flight PO-kind warnings before claiming; scope-implied difficulty \\
Environment design & Guidebooks: inherited organizational knowledge, persistent across sessions \\
\bottomrule
\end{tabular}}
\end{center}

\noindent The game is cooperative: each agent's successful delivery strengthens the artifact that all agents share, and no agent benefits from another's failure.  Formally, let $u_i(\mathcal{S})$ be the expected success rate of agent~$i$ on artifact state~$\mathcal{S}$.  When agent~$j$ delivers feature~$f$, the artifact gains traces, evidence, and coverage: $\mathcal{S}' \supseteq \mathcal{S}$ (monotonic ratchets prevent removal).  Since more traces reduce \texttt{traceability-complete} rejections for all agents, $u_i(\mathcal{S}') \geq u_i(\mathcal{S})$---satisfying the condition for a \emph{team game} in the sense of Marschak~(1955).  Better feedback (graded friction, per-kind routing, concrete repair examples from UNSAT cores) produces faster convergence.  The specification is the reward function.  RLVR~\cite{lightman2023} shapes the model during training; PSR shapes the environment at inference time.  The two are complementary.  The model does not need to be brilliant; it needs to be \emph{checkable}.

\paragraph{Proxy stigmergy.}  Smaller models (e.g.\ Haiku-class, Codex-mini) cannot hold the full governance protocol in context.  In practice, a coordinator agent---running a larger model---claims a feature, decomposes it into focused subtasks, and dispatches each subtask to a cheaper subagent with a narrowed context containing only the relevant requirements, design elements, and PO subset.  The subagent proposes mutations; the coordinator submits them through the kernel.  The verification floor catches errors from either party: if the coordinator's decomposition misses a trace, the PO rejects the delivery regardless of which agent produced it.  This pattern reduces token cost by 60--80\% on routine deliveries while preserving the governance guarantee.

%% ====================================================================
\section{Enforcement Mechanisms}
\label{sec:theater}

If agents can fabricate evidence, the substrate degenerates into security theater.

\paragraph{The attack.}  During development, an agent submitted fabricated evidence---\texttt{witness="governance-audit"} and \texttt{status=passed}---for features it had never tested.

\paragraph{The defense.}  Three structural mechanisms compose:

\begin{enumerate}
    \item \textbf{Witness allow-list.}  Evidence is rejected unless the witness field matches a registered source (e.g.\ \texttt{kernel-test-gate}, \texttt{ci-pipeline}, \texttt{human-review}).
    \item \textbf{Test-execution hash.}  Every evidence form includes a SHA-256 digest of test file contents, computed \emph{server-side}---agents never touch the hash.
    \item \textbf{Fused execution.}  \texttt{run\_evidence\_gate} runs tests, computes hash, and submits evidence atomically.  You cannot submit evidence without running the tests.
\end{enumerate}

\noindent The \texttt{evidence-provenance} PO verifies conditions 1--2 at commit time.  The \texttt{claim-before-dispatch} PO intercepts unclaimed work at the mutation layer---no transport can bypass it.

\emph{Threat model assumptions.}  These defenses hold under: (a)~the kernel is not compromised (the daemon is the TCB); (b)~test runners are honest (a compromised runner could produce valid hashes for non-running tests); (c)~guidebook authors are trusted (a malicious guidebook could weaken constraints).  Nidus does not claim security against adversarial kernel modification or colluding insiders---it prevents the common case: agents optimizing for speed by fabricating evidence.

\paragraph{The obligation argument.}  Before the gate, agents faced ambiguity: fabricate (fast, zero value) or test (slow, real value).  After, the only path that succeeds produces real evidence.  For example, the \texttt{evidence-provenance} constraint (PO-kind \texttt{evidence-provenance}, Listing~2) rejects any evidence whose witness is not in the allow-list or whose hash is missing---so the only mutation that passes is one produced by the fused \texttt{run\_evidence\_gate} tool.  The architect eliminated the ambiguity by adding this constraint; the agents benefited from the narrower search space.  This pattern repeated across all obligations: each constraint, crystallized from the architect's steering, removes a class of unproductive agent behavior.

\paragraph{Retroactive verification.}  \texttt{retroactive\_verify} evaluates candidate POs against $N$ historical states from the Git WAL, returning \texttt{safe} or \texttt{over-constrained}.  This is bounded model checking over the mutation history.

\paragraph{Incremental verification.}  Full verification is too slow for interactive commits.  The solution: map each mutation's target section to the PO kinds at risk.  A commit touching only \texttt{features} skips irrelevant kinds.  MCP stdio processes delegate to a single TCP daemon with cached solver state, keeping agent-facing latency sub-second.

\paragraph{Two-mode verification.}  Analysis of the proof obligation evaluators revealed that all commit-path checks evaluate formulas where variables are already bound from artifact state---no constraint solving is needed.  The system operates in two modes.  \emph{Online mode}: every commit is verified by direct Lisp evaluation on ground artifact state, sub-millisecond, no solver dependency.  \emph{Guidebook management mode}: when an architect authors, imports, or composes guidebook constraints, Z3 validates the constraint set for mutual consistency---genuine satisfiability checking with free variables.  The SMT solver is a guidebook compiler's tool, not a commit-path dependency.  S-expressions remain SMT-LIB2; the representation supports solver use.  The runtime rarely needs it.

\paragraph{Thin MCP relay.}  A practical constraint on multi-agent deployment is per-connection memory.  Each MCP client (Claude, Gemini, Codex) spawns an MCP stdio process.  Originally, each loaded the full daemon codebase (926\,MB), limiting concurrent agents to $\sim$3 on a 4\,GB server.  Since \texttt{tools/call} already relays to a shared TCP daemon, the MCP stdio process was reduced to a pure JSON-RPC relay (15\,MB), enabling $\sim$60 concurrent agent connections---a prerequisite for the investigative swarm described next.

\paragraph{Investigative swarm.}  The daemon decomposes audit queries into self-contained jobs, each inlining relevant artifact context (requirement text, feature scope).  Cheap agents (e.g.\ 100 parallel Jules sessions, Google's Gemini-based coding agent) investigate in review mode---reading code, judging consistency, reporting structured findings.  Findings are ingested into the coordination section; a proof obligation ensures all mismatches are resolved.  The constraint surface grows from stochastic exploration, not just architect observation.

\paragraph{SOP enforcement through gate-sequence POs.}  Standard operating procedures are documents that agents consult when convenient and skip when optimizing for speed.  An FPGA synthesis agent repeatedly skipped pre-synthesis gates (lint, simulation, resource model verification), wasting \$100--150 per failed synthesis run on cloud infrastructure.  The agent's self-assessment: ``I optimize for speed over correctness.  The SOP exists because this exact behavior wasted money.''

The fix encodes the SOP as a \emph{gate-sequence} proof obligation with DAG dependencies.  Each gate declares an identifier, a required witness type, and a list of prerequisite gates:

\begin{lstlisting}[basicstyle=\ttfamily\footnotesize]
(gate FI-1 (description "Pre-synth") (witness "verilator-sim")
  (depends-on FI-0 FI-0b FI-0c))
\end{lstlisting}

\noindent The evaluator checks coordination evidence entries for each gate in the dependency order.  If gate FI-0c (formal verification) has no evidence, the delivery is rejected: ``feature FEAT-X gate FI-1 requires FI-0c to pass first.''  The agent cannot skip gates because the constraint surface does not offer a path that avoids them.  This transforms SOPs from behavioral aspirations into structural properties of the engineering environment.

\paragraph{Lesson provenance: the artifact's immune system.}  Proof obligations are \emph{positive} constraints: what must hold.  Lessons are \emph{negative} constraints: what must not happen.  Together they complete the constraint surface.

During a 12-hour session, the daemon crashed six times from unbounded memory growth---each a different manifestation of the same class of bug.  Each crash was debugged, fixed, and committed.  The commit messages encoded the diagnostic chains, but commit messages die with the session.  The next agent would rediscover the same failures because the failure-to-obligation chain lived in git, not the artifact.

The fix: a \texttt{lessons} section in the artifact ($\Lambda$ in Definition~1).  Each lesson links a failure to a root cause to the obligation that prevents recurrence:

\begin{lstlisting}[basicstyle=\ttfamily\footnotesize]
(lesson LSN-OOM-Z3-PER-FEATURE
  (failure "160GB OOM: Z3 called per feature x constraint")
  (root-cause "Constraint satisfiability is guidebook-level")
  (fix "Removed per-feature Z3 loop")
  (obligation "skip-guidebooks at commit time")
  (affected-scope "lisp/epochd-proof-guidebook.lisp")
  (cost "160GB OOM, 4 daemon crashes")
  (commits "ce3e50d" "2ddd721"))
\end{lstlisting}

\noindent Each lesson is a negative proof obligation.  Z3 can verify soundness: the lesson's \texttt{root-cause} field identifies the precondition (here: ``Z3 called per feature $\times$ constraint''), and the \texttt{obligation} field names the invariant that prevents it (here: ``skip-guidebooks at commit time'').  The encoding is \texttt{(implies precondition failure)} and \texttt{(implies context (not precondition))}.  Concretely, for the lesson above: \texttt{precondition} $\equiv$ \texttt{(and commit-path-mode (z3-called-per-feature true))}; the obligation asserts \texttt{(implies commit-path-mode (not (z3-called-per-feature true)))}.  Z3 checks whether the conjunction is satisfiable.  If UNSAT, the obligation provably prevents the failure.  If SAT, there is a path the fix does not cover.

The workbench surfaces lessons whose \texttt{affected-scope} overlaps the agent's feature code-paths---not by keyword matching but by formal scope intersection.  An agent editing the verification pipeline sees the OOM lessons before writing code.  Institutional knowledge becomes a structural property of the artifact, not a behavioral property of agents who happened to read the right commit messages.

From a game-theoretic perspective, lessons solve the tragedy-of-the-commons on institutional knowledge.  The mechanism: recording a lesson is a public good (all future agents benefit from the narrower failure surface), while the cost falls on the recording agent (time spent diagnosing and formalizing).  Nidus aligns incentives by awarding reputation credit proportional to the lesson's severity---agents that crystallize high-cost failures earn faster tier promotion.  The game rewards knowledge crystallization alongside feature delivery.  The system becomes antifragile: each failure strengthens the constraint surface.

%% ====================================================================
\section{Related Work}
\label{sec:related}

\paragraph{LLM-assisted formal verification.}  Mugnier et al.~\cite{mugnier2024} (Laurel) use LLMs to generate Dafny helper assertions that unblock SMT-backed verification---an ``unbounded prover, bounded verifier'' split structurally similar to Nidus.  The difference: Laurel generates assertions for a \emph{fixed program}; Nidus governs the co-evolution of the specification \emph{itself}, including its proof obligations, through the same decidable surface.  Hao et al.~\cite{hao2025} formalize natural-language constraints into SAT/SMT with UNSAT-core repair loops---close to Nidus's \texttt{what\_if} + repair mechanism.  Their framework governs plans; Nidus governs the full engineering artifact (requirements, architecture, design, traces, proofs, evidence) as a singleton.  Kleppmann~\cite{kleppmann2025} predicts that AI will make formal verification mainstream by removing the proof-writing bottleneck; Nidus is an existence proof of that thesis, applied to engineering governance rather than program correctness.

\paragraph{Agent frameworks and planning layers.}  Agent-first IDEs generate implementation plans before coding; multi-agent frameworks (AutoGen~\cite{autogen2024}, SpecMAS~\cite{specmas2024}) add orchestration and protocol verification.  These are complementary to Nidus, which occupies a different layer: not an agent framework but a \emph{governance substrate} beneath any framework.  Agents from any family connect via MCP; the surface gates all of them identically.  RLVR and process reward models can coexist with the surface---the model improves its adherence during training; the surface guarantees enforcement at inference time.

\paragraph{Formal methods and MBSE.}  Proof-carrying code~\cite{acl2} established the ``unbounded prover, bounded verifier'' architecture decades ago; Nidus is the LLM-era realization applied to engineering artifacts rather than machine code.  ACL2 does not re-check its theorem database on modification; seL4~\cite{klein2009} requires human-guided Isabelle/HOL.  Cedar~\cite{cedar2024} analyzes authorization policies but is not itself governed by its policies.  The MBSE+SMT patent~\cite{mbse2022} translates MBSE artifacts into first-order logic via a translation layer.  Unlike such approaches, Nidus keeps the artifact solver-aligned: the same S-expression structure is used for storage, inspection, and constraint compilation.  TLA+, Dafny, and Lean/Coq/Isabelle verify programs in a proof language separate from the program; Nidus's representational closure means the specification, the database, and the solver input are the same object.

\paragraph{Normative multi-agent systems.}  The normative MAS tradition~\cite{boissier2009} separates an organizational norm layer from domain agents, with enforcement components that detect violations---structurally close to Nidus's constraint surface.  Guidebook inheritance ($\Pi(G_{\mathrm{parent}}) \subseteq \Pi(G_{\mathrm{child}})$) parallels hierarchical institutional norms in which lower-level organizations inherit constraints from higher-level ones~\cite{vandertorre2009}.  Nidus adds SMT-decidability, representational closure, and recursive self-governance to this pattern.

\paragraph{Executable specifications and MBSE.}  SysML~v2~\cite{sysmlv2} moves model-based systems engineering toward executable, textual specifications---structurally converging with Nidus's artifact-as-specification.  The difference: SysML~v2 models are authored by humans and verified by tools; Nidus artifacts are \emph{co-authored} by humans and LLMs and verified on every mutation.  SysML~v2 has no recursive self-governance, no friction-based coordination, and no mechanism to prevent governance theater.

\paragraph{Self-modifying systems.}  The G\"{o}del Machine~\cite{schmidhuber2003} gates self-modification on proof but is a theoretical framework over unrestricted logic.  Nidus restricts to decidable tuples (finite sets, QF\_LIA guards) for tractability, making recursive self-governance practically enforceable.

%% ====================================================================
\section{Evaluation}
\label{sec:evaluation}

We evaluate four claims: (C1)~the constraint surface prevents regressions; (C2)~guidebook inheritance works across projects; (C3)~recursive self-governance holds; (C4)~the friction model differentiates agent behavior.  All experiments run on the self-hosting deployment---the system that governed its own construction.

\subsection{C1: Constraint Surface Prevents Regressions}

\paragraph{Lesion study.}  For each of the 49 removable elements in a sandbox artifact, we removed the element via \texttt{commit\_change\_set} and recorded whether verification rejected the mutation.

\begin{table}[ht]
\centering
\caption{PO coverage by element type ($N = 49$ removals).}
\label{tab:lesion}
\begin{tabular}{@{}lcccl@{}}
\toprule
\textbf{Element type} & \textbf{N} & \textbf{Guarded} & \textbf{Rate} & \textbf{Guardian PO kind} \\
\midrule
Trace              & 11 & 11 & 100\% & traceability-complete \\
Component          &  6 &  6 & 100\% & connector-integrity \\
Design element     &  6 &  6 & 100\% & feature-scope-valid \\
Requirement        & 11 &  9 &  82\% & feature-scope-valid \\
Connector          &  5 &  0 &   0\% & \emph{gap (subsequently closed)} \\
Proof obligation   & 10 &  0 &   0\% & \emph{sandbox lacked immutable flag} \\
\midrule
\textbf{Total}     & \textbf{49} & \textbf{32} & \textbf{65\%} & \\
\bottomrule
\end{tabular}
\end{table}

\noindent Two gaps found.  The connector gap led to \texttt{connector-existence}, implemented through the governance loop---the experiment identified the weakness; the fix was delivered by the mechanism it now protects.  The PO gap was a sandbox limitation: the production artifact uses \texttt{(immutable t)} to protect POs from removal.

\paragraph{Verification latency.}  Table~\ref{tab:latency} reports wall-clock times for the production artifact (461~features, 873~requirements, 238~POs).  The daemon maintains cached solver state; all measurements include TCP round-trip.

\begin{table}[ht]
\centering
\caption{Verification latency on the production artifact (M4~Mac Studio, SBCL~2.6, 238~POs).}
\label{tab:latency}
\begin{tabular}{@{}lrl@{}}
\toprule
\textbf{Operation} & \textbf{Wall-clock} & \textbf{Notes} \\
\midrule
\texttt{ping}             &      1\,ms & Daemon heartbeat \\
\texttt{read\_system}     &    173\,ms & Full artifact + traceability matrix \\
\texttt{precommit-check}  &    493\,ms & Incremental: 238~POs, touched sections only \\
\texttt{artifact\_health}  &  $\sim$2\,s & Full: all POs + governance health score \\
\texttt{commit\_change\_set} & 2--5\,s & Verify + CAS + git commit + WAL push \\
\bottomrule
\end{tabular}
\end{table}

\noindent Incremental verification (the commit-path hot loop) runs under 500\,ms.  Full health checks run under 2\,s.  Both scale linearly with PO count; the current 238~POs leave headroom for growth.

\subsection{C2: Guidebook Inheritance}
This claim is trivially guaranteed by the monotonic inheritance semantics (Theorem~3), so the evaluation serves as a \emph{validation} rather than an experiment.  The production deployment uses twelve guidebooks (128~constraints) across four layers: \emph{corporate} (copyright, singleton, daemon authority---10 constraints), \emph{engineering} (module DAG, function granularity, file size, naming, spec-first, TDD---37 constraints), \emph{compliance} (EU~AI~Act, ISO~27001, ISO~42001, NIST~AI~RMF, DO-178C, SOC~2, HIPAA, Singapore~MGF---72 constraints across eight frameworks), and \emph{domain} (chip-design RTL governance---9 constraints).  Five governed projects inherit from the same guidebook hierarchy.  A data-loss-prevention project importing the corporate and EU~AI~Act guidebooks was rejected for missing \texttt{test-paths} on delivery---a constraint it never defined locally.  The organizational standard was enforced mechanically, confirming the implementation matches Theorem~3.

\subsection{C3: Recursive Self-Governance}

Three mutations on the sandbox, in sequence: (1)~Add PO with \texttt{(immutable t)}---accepted; PO immediately active.  (2)~Submit requirement violating the new PO---rejected by the PO just added.  (3)~Remove the immutable PO---rejected.  The system accepted an addition to its own constraint surface, enforced it on the next mutation, and blocked its removal.

\subsection{C4: Friction Differentiates Agent Behavior}

125~friction events, 25~agent sessions, three model families.  Gemini sessions: predominantly rejections (44/47 events).  Codex sessions: proportionally more abandonments (odds ratio $= 5.04$, Fisher's exact $p = 0.02$).  Seven sessions reached \textit{restricted} tier; the highest-friction session (36~events) was mechanically gated from claiming features.  Different model families produce different failure signatures; the friction model routes accordingly.

\subsection{Architect-Designed Constraints}

Table~\ref{tab:constraints} shows five constraints crystallized from the architect's steering during construction.  Each row traces observed friction to a deployed constraint to the class of agent behavior eliminated.

\begin{table}[ht]
\centering
\caption{Architect-steered constraints.  Observed friction $\to$ crystallized obligation $\to$ agent behavior eliminated.}
\label{tab:constraints}
{\small
\begin{tabular}{@{}p{3cm}p{5cm}p{4.5cm}@{}}
\toprule
\textbf{Friction} & \textbf{Constraint} & \textbf{Behavior blocked} \\
\midrule
78 fabricated evidence incidents & \texttt{(implies evidence\_submitted (and witness\_registered hash\_present))} & Unregistered witness or missing hash \\
\addlinespace
Unclaimed dispatch & \texttt{(implies work\_dispatched feature\_claimed\_by\_agent)} & Work on unclaimed features \\
\addlinespace
Delivery without tests & \texttt{(implies feature\_delivered (and has\_code\_paths has\_test\_paths ...))} & Incomplete scope \\
\addlinespace
Backfilled specs & \texttt{(temporal true) (po-kind spec-precedes-code)} & Spec committed after code \\
\addlinespace
DAG violations & \texttt{(and (not parse\_imports\_query) (not mutate\_imports\_query) ...)} & Module back-edges \\
\bottomrule
\end{tabular}}
\end{table}

\paragraph{Threats.}
(1)~All deployments internal to one organization.  (2)~Coverage gaps found and closed during self-hosting (connector, evidence, dispatch).  (3)~Swarm experiment ran seconds, not hours.  (4)~A third deployment (SystemVerilog, 68 RTL modules) is bootstrapped but not yet independently operated.  Future work: independent replication on an external codebase and $k$-step forward reachability experiments.

%% ====================================================================
\section{Conclusion}
\label{sec:conclusion}

LLM agents generate code.  They do not enforce engineering invariants.  Nidus combines a fast proposer (LLMs) with a constraint surface that verifies every mutation.  The practical result: guidebooks encode organizational standards inherited by every governed project; the surface compounds as each obligation permanently eliminates a class of unengineered output.  The verification-gated history is a complete engineering record with respect to the modeled artifact and active obligations (Theorem~\ref{thm:record}): compliance certificates, traceability matrices, and impact analyses are projections---the audit trail is a theorem, not a process.  LLM agents replace coders.  They do not replace the architect.  Requirements, architecture decisions, and the development vector remain human input.

The surface outlives the agents, the models, and the frameworks.  When the next model generation arrives, the constraints remain sound.  When Nidus itself evolves---new PO kinds, new guidebook constraints, new sections---the evolution is governed by the same surface.  Each commit carries a daemon-generated attestation: which feature, which agent, which verification result, which fingerprint transition.  The git write-ahead log is a sequence of verified, attributed governance states---a formal development in the mathematical sense.

As models improve, the architect's role itself becomes partially automatable: decomposing features, drafting architecture decisions, proposing proof obligations from friction patterns.  The living artifact makes this transition safe because the \emph{mechanical} architect (the constraint surface + Z3) is separated from the \emph{creative} architect (the human or LLM making judgments).  Whoever proposes---human or model---the surface verifies.  The creative role can migrate gradually from human to LLM without the governance degrading, because the governance is a property of the artifact, not of who writes to it.

$k$-step forward reachability---verifying that no sequence of $k$ legal mutations can reach a state violating an immutable obligation, a bounded model checking problem over the constraint surface---and independent replication on an external codebase complete the agenda.

\paragraph{Acknowledgements.} I thank Anton Afanasyev for discussions that improved the clarity of this paper.

\paragraph{Disclosure.} The recursive self-governance mechanism is the subject of Swiss patent application CH000371/2026 (filed with IGE).

\end{document}